\documentclass{article}

\usepackage[utf8]{inputenc}
\usepackage[UKenglish]{babel}
\usepackage{amsmath, amsfonts, amssymb}
\usepackage{graphicx, subcaption}
\usepackage{amsthm}

\title{An effective construction for cut-and-project rhombus tilings with global $n$-fold rotational symmetry}

\author{Victor H. Lutfalla}

\newtheorem{theorem}{Theorem}
\newtheorem{proposition}{Proposition}
\newtheorem{corollary}{Corollary}
\newtheorem{lemma}{Lemma}

\newcommand{\imag}{\mathrm{i}}

\newcommand{\openinterval}[2]{ \left] #1 ,\, #2 \right[ } 
\newcommand{\closedinterval}[2]{\left[ #1,\, #2 \right]} 
\newcommand{\closedopeninterval}[2]{\left[ #1,\, #2 \right[}

\newcommand{\grid}[2]{H(#1,#2)}
\newcommand{\real}[1]{\mathrm{Re}\left(#1\right)}
\newcommand{\multigrid}[2]{G_{#1}(#2)}
\newcommand{\dualtiling}[2]{P_{#1}(#2)}
\newcommand{\edge}[2]{\{#1,#2\}}

\newcommand{\ie}{\emph{i.e.} }

\renewcommand{\leq}{\leqslant}
\renewcommand{\geq}{\geqslant}

\newcommand{\floor}[1]{\left\lfloor #1 \right\rfloor}
\newcommand{\ceil}[1]{\left\lceil #1 \right\rceil}

\begin{document}

\maketitle


\paragraph{Acknowledgements:}I want to thank Emmanuel Jeandel for pointing me towards the article \emph{Trigonometric diophantine equations (On vanishing sums of roots of unit)} \cite{conway1976} when I presented him with the problem of regularity of multigrids. I also want to thank my advisor Thomas Fernique and my colleagues Lionel Pournin and Thierry Monteil for their help and advice.

\begin{abstract}
We give an explicit and effective construction for rhombus cut-and-project tilings with global $n$-fold rotational symmetry for any $n$. This construction is based on the dualization of regular $n$-fold multigrids. The main point is to prove the regularity of these multigrids, for this we use a result on trigonometric diophantine equations.
\end{abstract}

\section{Introduction}
\label{sec:introduction}
In the 70s Penrose presented a quasiperiodic rhombus tiling with global 5-fold rotational symmetry and local 10-fold rotational symmetry \cite{penrose1974}. This tiling is one of the most thoroughly studied tilings, for more details on the class of Penrose rhombus tilings and on the canonical Penrose rhombus tilings see \cite{senechal1996,baake2013}. In 1981 de Bruijn proposed an algebraic definition of this tiling by the dualization of a pentagrid \cite{debruijn1981}. Later on a generalized multigrid method has been shown to be equivalent to the cut-and-project method or projection method for the construction of tilings by Gähler and Rhyner \cite{gahler1986}. Here we use this dualization of multigrid method to give an effective construction of cut-and-project rhombus tilings with global  $2n$-fold rotational symmetry for any $n$, and a similar construction with global $n$-fold rotational symmetry for odd $n$.

The main point of our work is to prove the regularity of the multigrids involved because when the multigrid is not regular its dual tiling contains tiles that are not rhombuses as depicted in Figure \ref{fig:multigrid_intersection}.
Note that, by a non-constructive cardinality argument, the existence of such tilings has been known since the multigrid method was introduced and even though some might consider it known as \emph{folk} we were not able to find any reference for explicit construction of regular grids with global $n$-fold rotational symmetry outside of the case of pentagrids \cite{debruijn1981} and tetragrids \cite{beenker1982}. 
This explicit construction is used as an intermediate step for the construction of substitution discrete planes with $n$-fold rotational symmetry in \cite{kari2020substitution}.

Let us now give the relevant definitions for rhombus tilings and for the multigrid method before stating the main results.

A \emph{rhombus tiling} is a covering of the plane without overlap by a set of rhombus tiles. Here we actually consider rhombus tilings of the complex plane $\mathbb{C}$ and we only consider the case where the set of tiles is finite up to translation and where the tiling is \emph{edge-to-edge} \ie any two tiles in the tiling either do not intersect, intersect on a single common vertex or intersect along a full common edge. A \emph{patch} is a finite simply-connected set of non-overlapping tiles and a \emph{pattern} is a patch up to translation. 

A tiling is called \emph{periodic} of period $\vec{v}\neq \vec{0}$ when it is invariant under the translation of vector $\vec{v}$ and \emph{non-periodic} when it admits no period.
A tiling is called \emph{uniformly recurrent} or \emph{uniformly repetitive} when for any pattern $p$ that appears in the tiling, there exists an ``appearance'' radius $R$ such that in the intersection of any open ball of radius $R$ with the tiling there is an occurrence of the pattern $p$.
A tiling is called \emph{quasiperiodic} when it is both non-periodic and uniformly recurrent.

A tiling is said to have \emph{global $n$-fold rotational symmetry} when there exists a point $z$, usually the origin, such that the tiling is invariant under the rotation of center $z$ and angle $\tfrac{2\pi}{n}$.
The \emph{crystallographic restriction} theorem states that the only rotational symmetries possible for periodic tilings are 2-fold, 3-fold, 4-fold and 6-fold. Hence a tiling that has $n$-fold rotational symmetry for $n\notin\{2,3,4,6\}$ is non-periodic.

\emph{Cut-and-project} tilings are tilings that can be seen as the projection of a discrete plane of some $\mathbb{R}^n$ to the plane. The tiles of a cut-and-project tiling are the projection of the facets that form the corresponding discrete plane. We will not give here a precise definition since we will not use it later, for precise definitions see \cite{baake2013} or \cite{senechal1996}. Note that cut-and-project tilings are used to model quasicrystals which means that they are not only of interest to mathematicians and theoretical computer-scientists but also to physicists and crystallographists. The only result we will use on cut-and-project tilings is that they are uniformly recurrent \cite{senechal1996}, in particular this implies that cut-and-project tilings with $n$-fold rotational symmetry for $n\notin\{2,3,4,6\}$ are quasiperiodic.

\begin{figure}[t]
    \begin{subfigure}{0.45\textwidth}
      \includegraphics[width=\textwidth]{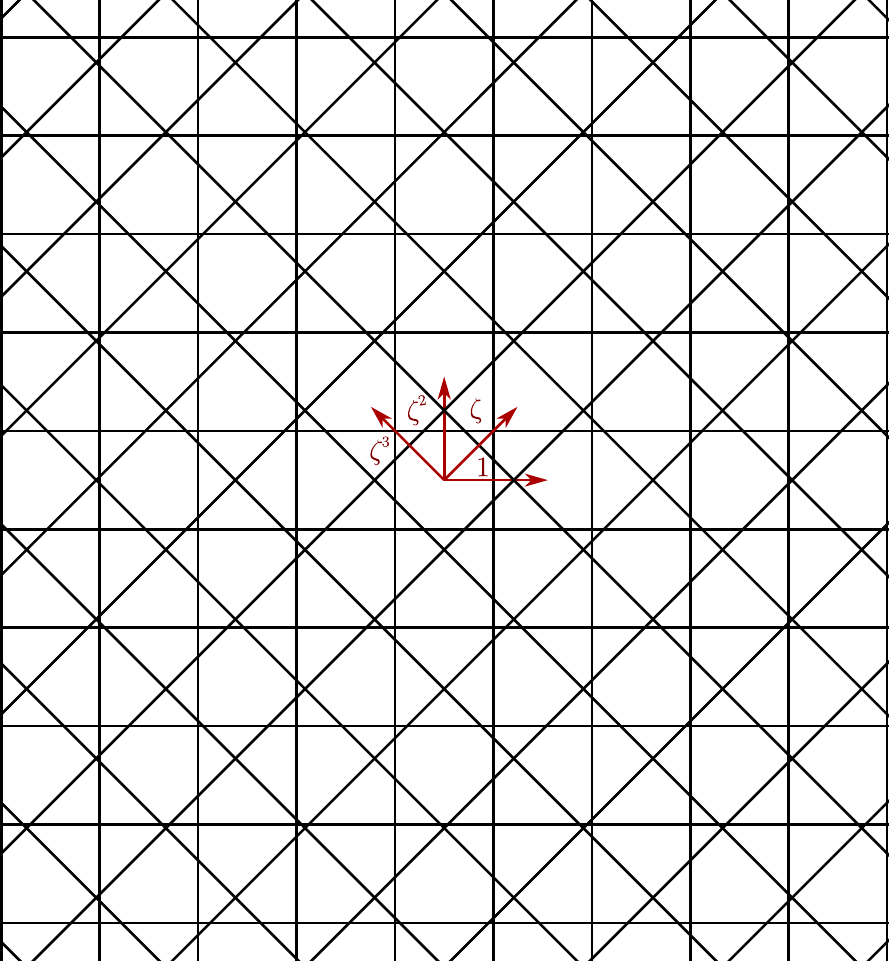}
      \caption{$\multigrid{4}{\tfrac{1}{2}}$}
      \label{fig:multigrid_4}
    \end{subfigure}
    \qquad
    \begin{subfigure}{0.45\textwidth}
      \includegraphics[width=\textwidth]{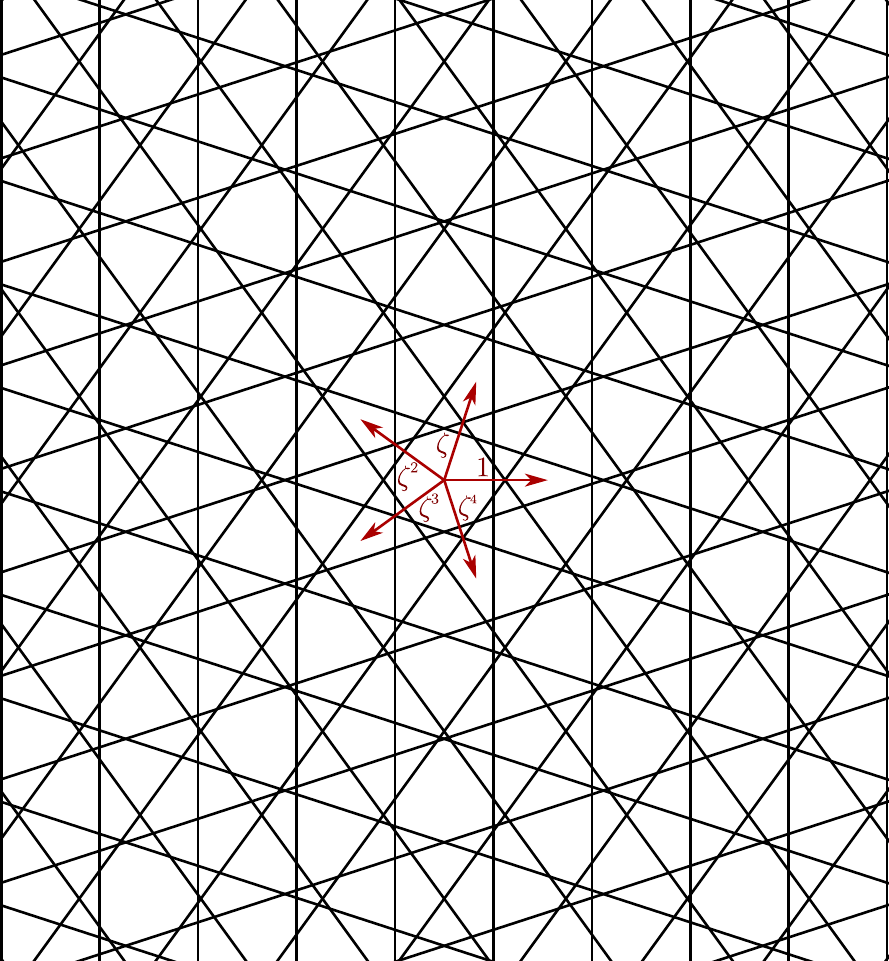}
      \caption{$\multigrid{5}{\tfrac{1}{2}}$}
      \label{fig:multigrid_5}
    \end{subfigure}
    
    \caption{Examples of multigrids}
    \label{fig:multigrid_4_5}
  \end{figure}

Let us now define grids, multigrids and their dual tilings.
The \emph{grid} of direction $\xi\in\mathbb{C}$ with $|\xi| = 1$ and offset $\gamma\in\mathbb{R}$, denoted by $\grid{\xi}{\gamma}$, is the set of equidistant lines orthogonal to $\xi$ and with offset $\gamma$ from the origin
\[ \grid{\xi}{\gamma} :=  \left\{ z \in \mathbb{C}\ |\ \real{z\cdot \bar{\xi}}-\gamma \in \mathbb{Z}\right\}.\] In this definition the important offset is actually the fractional part of $\gamma$ so we restrict the definition to the case $0\leq \gamma < 1$.
The \emph{multigrid} of pairwise non-collinear directions $\xi = (\xi_i)_{0\leq i < n}$ and offsets $\gamma = (\gamma_i)_{0\leq i < n}$, denoted by $\multigrid{\xi}{\gamma}$ is the union of the grids
\[\multigrid{\xi}{\gamma}:=\bigcup\limits_{0\leq i <n} \grid{\xi_i}{\gamma_i}.\]

For an integer $n\geq 3$ we define the $n$-fold direction $\zeta_n$ as
\[ \zeta_n = \begin{cases} e^{\imag\tfrac{2\pi}{n}} \text{ if } n \text{ is odd}\\  e^{\imag\tfrac{\pi}{n}} \text{ if } n \text{ is even} \end{cases} \]

The \emph{$n$-fold multigrid} of offsets $\gamma = (\gamma_i)_{0\leq i < n}$ denoted $\multigrid{n}{\gamma}$ is the multigrid with directions $(\zeta_n^i)_{0\leq i < n}$
\[\multigrid{n}{\gamma} := \bigcup\limits_{0\leq i <n} \grid{\zeta^i_n}{\gamma_i} \]
For simplicity for some real number $x\in\closedopeninterval{0}{1}$ we denote $\multigrid{n}{x}$ the multigrid with all offsets equal to $x$ \ie
\[ \multigrid{n}{x} := \multigrid{n}{(x,\dots x)} = \bigcup\limits_{0\leq i <n} \grid{\zeta^i_n}{x}\]

See Figure \ref{fig:multigrid_4_5} for a picture of $\multigrid{4}{\tfrac{1}{2}}$ and $\multigrid{5}{\tfrac{1}{2}}$.

From a multigrid $\multigrid{\xi}{\gamma}$ we can define a tiling $\dualtiling{\xi}{\gamma}$ by a duality process.  This tiling is dual to the multigrid in the sense that the multigrid, seen as a graph, is the adjacency graph of the tiling, see Figure \ref{fig:example_grid_tiling} for an example. We will define this tiling only in the $n$-fold case, but one can easily adapt it to the general case by replacing the directions $\zeta_n^i$ by $\xi_i$.

\begin{figure}[t]
\center
\begin{subfigure}[b]{0.45\textwidth}
\includegraphics[width=\textwidth]{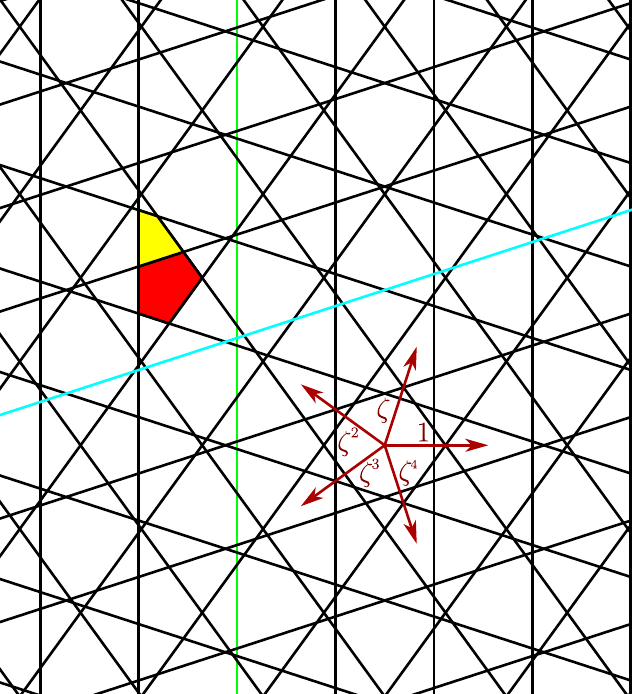}
\caption{The 5-fold multigrid or pentagrid $G_5(\tfrac{1}{2})$}
\label{fig:G5}
\end{subfigure}
\qquad
\begin{subfigure}[b]{0.45\textwidth}
\includegraphics[width=\textwidth]{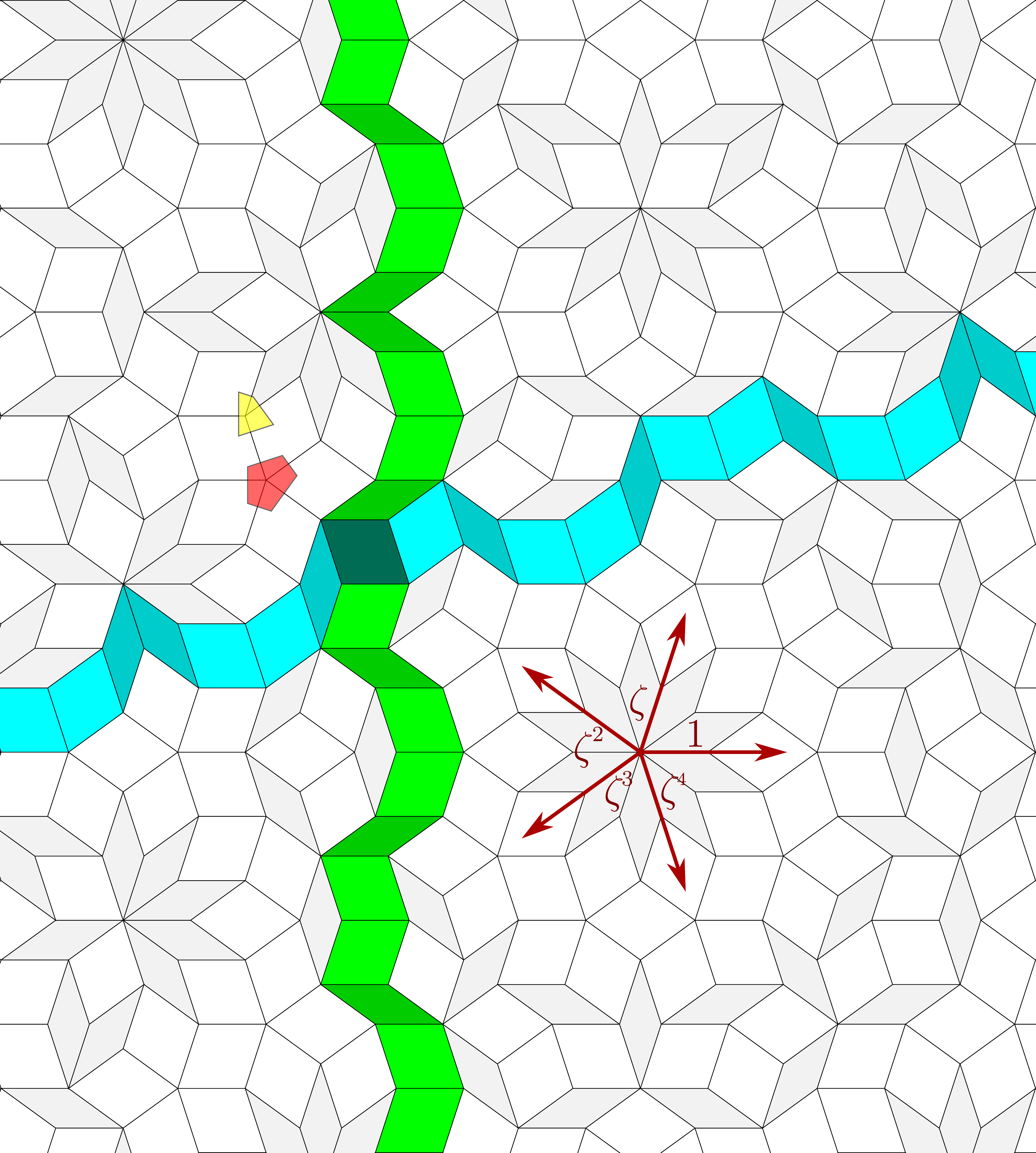}
\caption{The dual tiling $P_5(\tfrac{1}{2})$}
\label{fig:P5}
\end{subfigure}
\caption{Example of a regular grid and its dual tiling, some elements of the multigrid and their dual in the tiling have been colored.}
  \label{fig:example_grid_tiling}
\end{figure}

Given a $n$-fold multigrid $\multigrid{n}{\gamma}$ we define two functions $K: \mathbb{C} \to \mathbb{Z}^n$ and $f:\mathbb{C}\to \mathbb{C}$ as
\[ K(z):= \left( \left\lceil \text{Re}\left(z\cdot \bar{\zeta_n^i}\right)-\gamma_i \right\rceil\right)_{0\leq i <n}  \qquad f(z):= \sum\limits_{i=0}^{n-1}\left\lceil \text{Re}\left(z\cdot \bar{\zeta_n^i}\right)-\gamma_i \right\rceil \zeta_n^i.\]
Remark that the functions $K$ and $f$ are constant on the interior of each cell, also called mesh, of the multigrid, so $f$ sends a cell of the multigrid to a single vertex.

The \emph{multigrid dual tiling} of a $n$-fold multigrid $\multigrid{n}{\gamma}$, denoted by $\dualtiling{n}{\gamma}$, is defined by its set of vertices $V$ and of edges $E$ as
\[ V := f(\mathbb{C}) \qquad E:= \left\{ \edge{z}{z'},\ z,z' \in V |\ \exists i, z' = z + \zeta_n^i\right\}\]  

In this dualization process each cell or mesh of the multigrid is sent to a vertex of the dual tiling, see for example the cells in red and yellow and their dual in Figure \ref{fig:example_grid_tiling}, and each intersection point of the multigrid is sent to a tile of the dual tiling. The dual of an intersection point where $k$ lines intersect is a $2k$-gon with unit sides as shown in Figure \ref{fig:multigrid_intersection} for the case of 5-fold multigrids. The dual of a line of the multigrid is a chain or ribbon of tiles that share an edge, see for example the green and blue line and their dual in Figure \ref{fig:example_grid_tiling}.
Recall that multigrid dual tilings are cut-and-project tilings \cite{senechal1996,gahler1986,baake2013}.

\begin{figure}[t]
    \includegraphics[width=\textwidth]{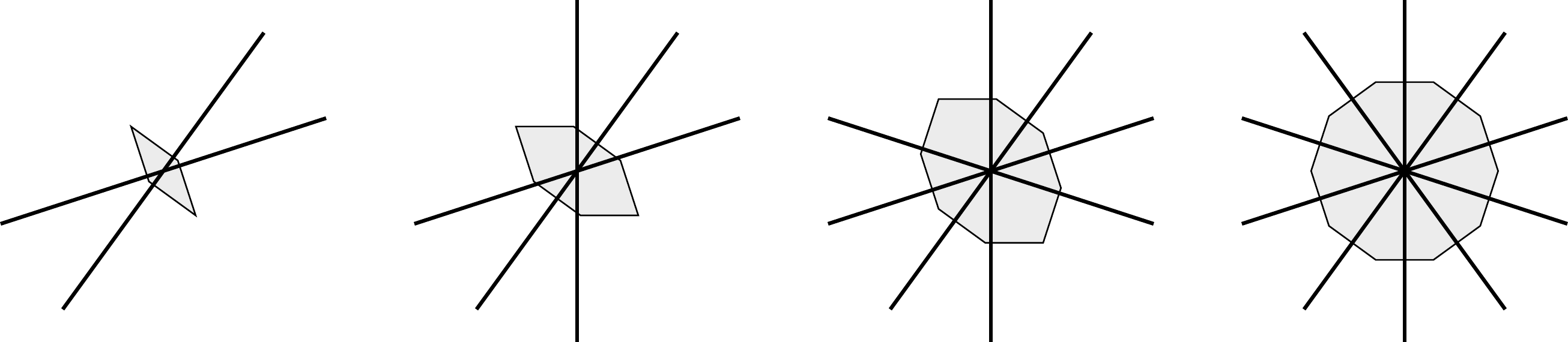}
    \caption{Possible intersection points in $\multigrid{5}{\gamma}$ and their dual tiles }
    \label{fig:multigrid_intersection}
\end{figure}

A multigrid is called \emph{singular} when there is at least one intersection point where at least 3 lines intersect, and is called \emph{regular} otherwise. By definition, the dual tilings of regular multigrids are edge-to-edge rhombus tilings.
For an example of a regular multigrid and its dual tiling, see Figure \ref{fig:example_grid_tiling} where some elements of the multigrid and their dual are highlighted to emphasize the dualization process.

Now that we have defined all relevant terms, we can state the main result.

\begin{theorem}~
  \begin{enumerate}
  \item for any integer $n\geq 4$, the $n$-fold multigrid dual tiling $\dualtiling{n}{\tfrac{1}{2}}$ is a cut-and-project quasiperiodic edge-to-edge rhombus tiling with global $2n$-fold rotational symmetry.
  \item for any \textbf{odd} integer $n\geq 5$, the $n$-fold multigrid dual tiling $\dualtiling{n}{\tfrac{1}{n}}$ is a cut-and-project quasiperiodic edge-to-edge rhombus tiling with global $n$-fold rotational symmetry. 
  \end{enumerate}
  \label{thm:tilings}
\end{theorem}

Theorem \ref{thm:tilings} is actually a corollary of Theorem \ref{thm:regularity} on the regularity of $n$-fold multigrids. See Figures \ref{fig:nfold_789} and \ref{fig:P23} for examples of tilings $\dualtiling{n}{\tfrac{1}{2}}$ and $\dualtiling{n}{\tfrac{1}{n}}$.

\begin{theorem}~
  \begin{enumerate}
  \item for any $n\geq 3$ and any non-zero rational offset $r\in\mathbb{Q}\cap \openinterval{0}{1}$ the $n$-fold multigrid $\multigrid{n}{r}$ is regular
  \item for any \textbf{odd} $n\geq 3$ and any tuple of non-zero rational offsets $\gamma = (\gamma_i)_{0\leq i < n} \in \left(\mathbb{Q}\cap \openinterval{0}{1}\right)^n$ the $n$-fold multigrid $\multigrid{n}{\gamma}$ is regular
  \end{enumerate}
  \label{thm:regularity}
\end{theorem}

Remark that Theorem \ref{thm:regularity} is not an exact characterization of regular $n$-fold multigrid but rather an easily checked sufficient condition for regularity. In \cite{debruijn1981} N. G. de Bruijn gives an exact characterization of regular pentagrids in the specific Penrose case (sum of the offsets is an integer), but this characterization is not easily generalized to the non-Penrose case and to all $n$-fold multigrids.

In Section \ref{sec:regular_tilings} we show how Theorem \ref{thm:tilings} is a corollary of Theorem \ref{thm:regularity}.
In Section \ref{sec:regular_trigonometric} we present the link between the regularity of multigrids and some trigonometric equations.
In Section \ref{sec:trigonometric_diophantine} we present the results of Conway and Jones on trigonometric diophantine equations \cite{conway1976}.
In Section \ref{sec:proof} we combine the results of Sections \ref{sec:regular_trigonometric} and \ref{sec:trigonometric_diophantine} to prove Theorem \ref{thm:regularity}.

\section{From regular $n$-fold multigrids to tilings with global $n$-fold symmetry}
\label{sec:regular_tilings}
\begin{figure}[t]
\center
\begin{subfigure}[b]{0.3\textwidth}
\includegraphics[width=\textwidth]{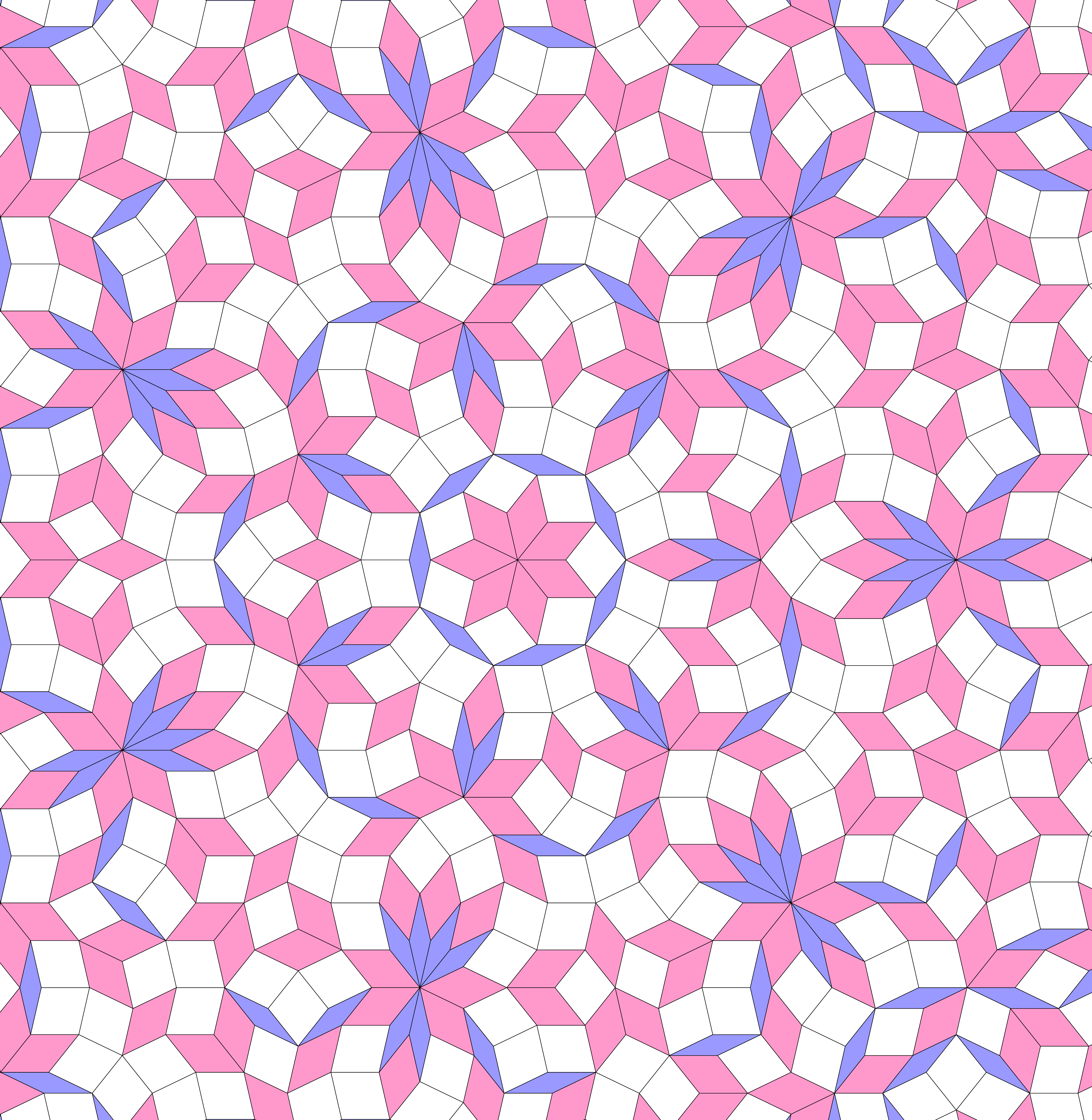}
\caption{$7$-fold : $\dualtiling{7}{\tfrac{1}{7}}$ }
\label{fig:P7}
\end{subfigure}
~
\begin{subfigure}[b]{0.3\textwidth}
\includegraphics[width=\textwidth]{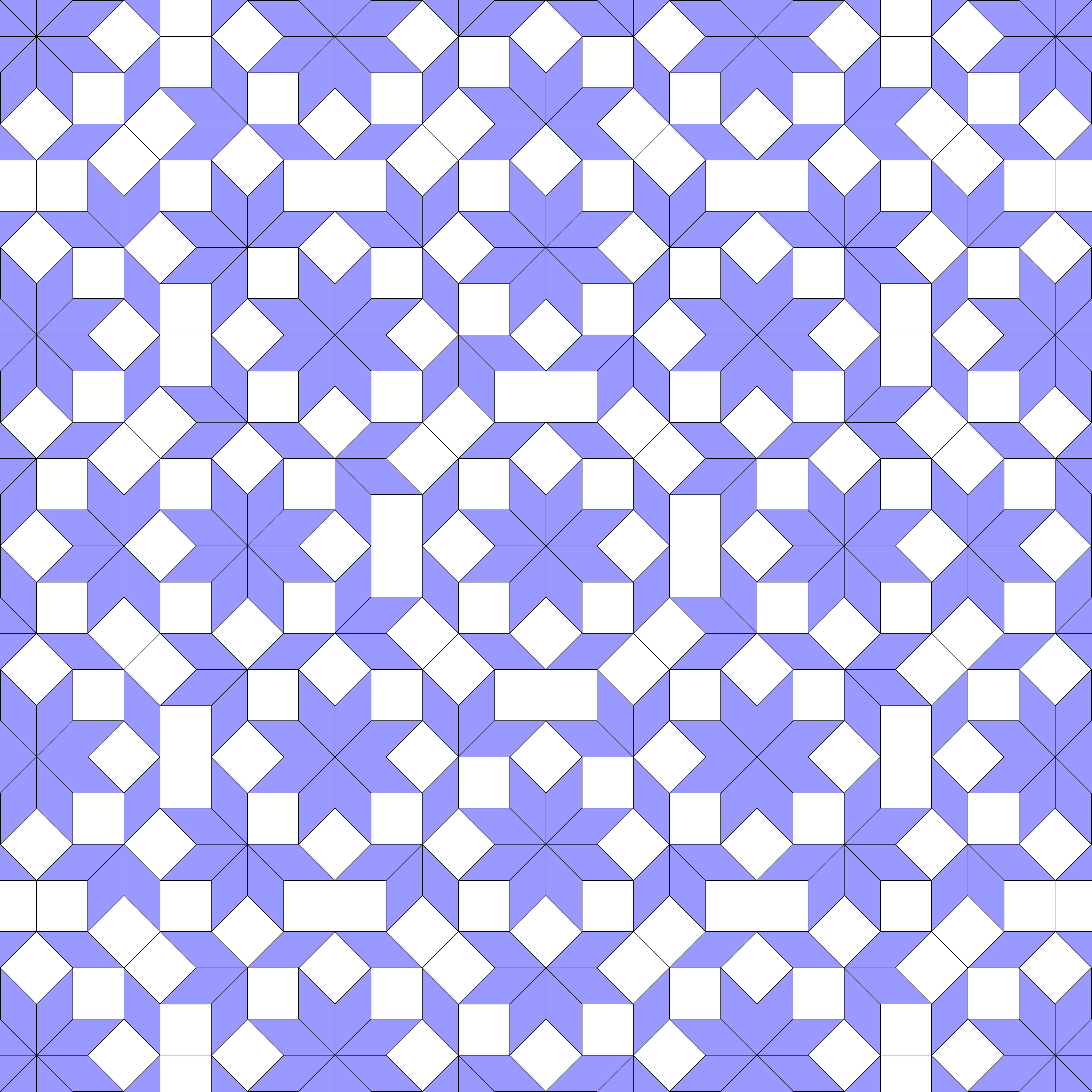}
\caption{$8$-fold : $\dualtiling{4}{\tfrac{1}{2}}$}
\label{fig:P4-2}
\end{subfigure}
~
\begin{subfigure}[b]{0.3\textwidth}
\includegraphics[width=\textwidth]{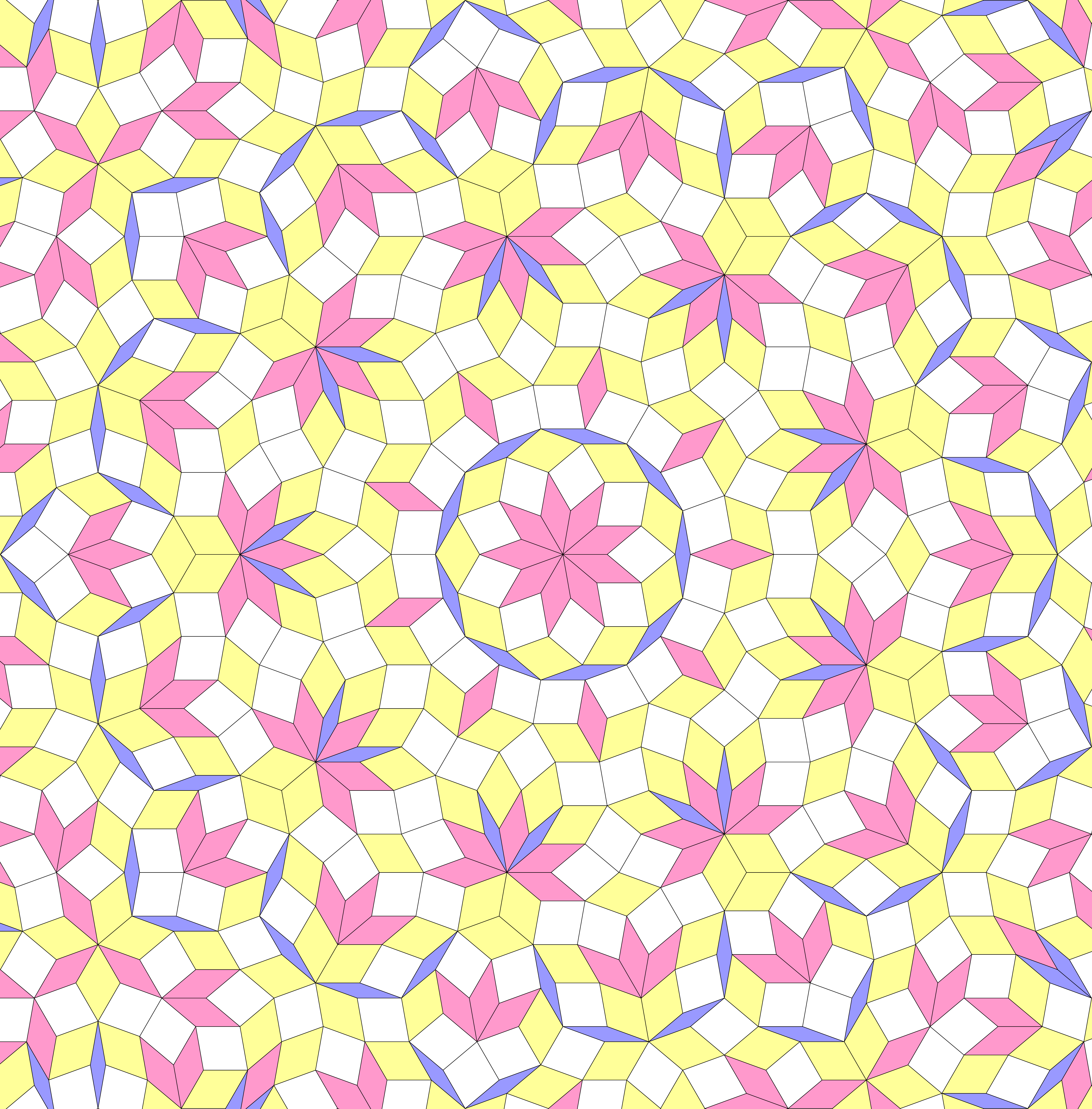}
\caption{$9$-fold : $\dualtiling{9}{\tfrac{1}{9}}$}
\label{fig:P9}
\end{subfigure}

~
\begin{subfigure}[b]{0.29\textwidth}
\includegraphics[width=\textwidth]{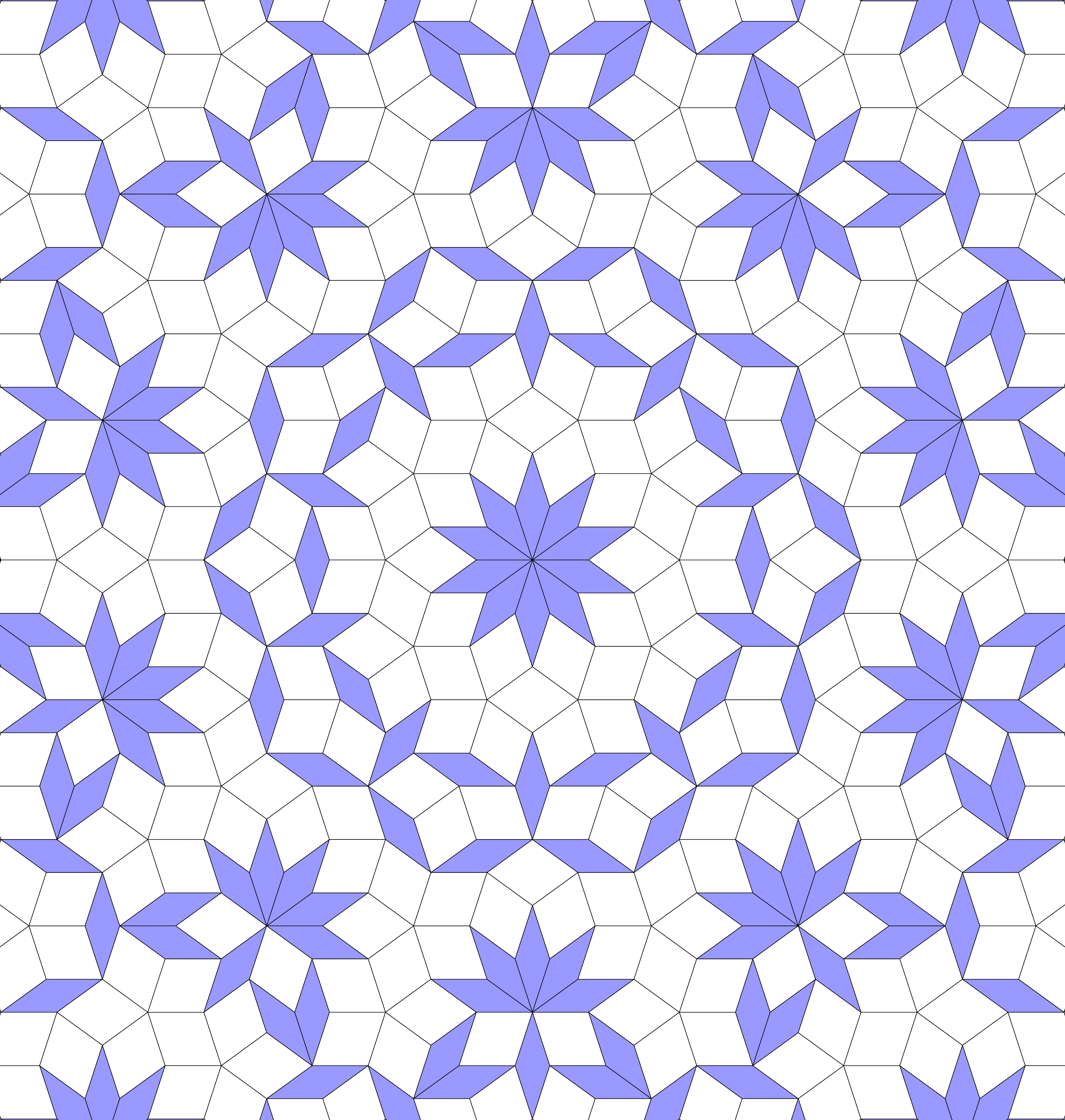}
\caption{$10$-fold : $\dualtiling{5}{\tfrac{1}{2}}$}
\label{fig:P5-2}
\end{subfigure}
~
\begin{subfigure}[b]{0.3\textwidth}
\includegraphics[width=\textwidth]{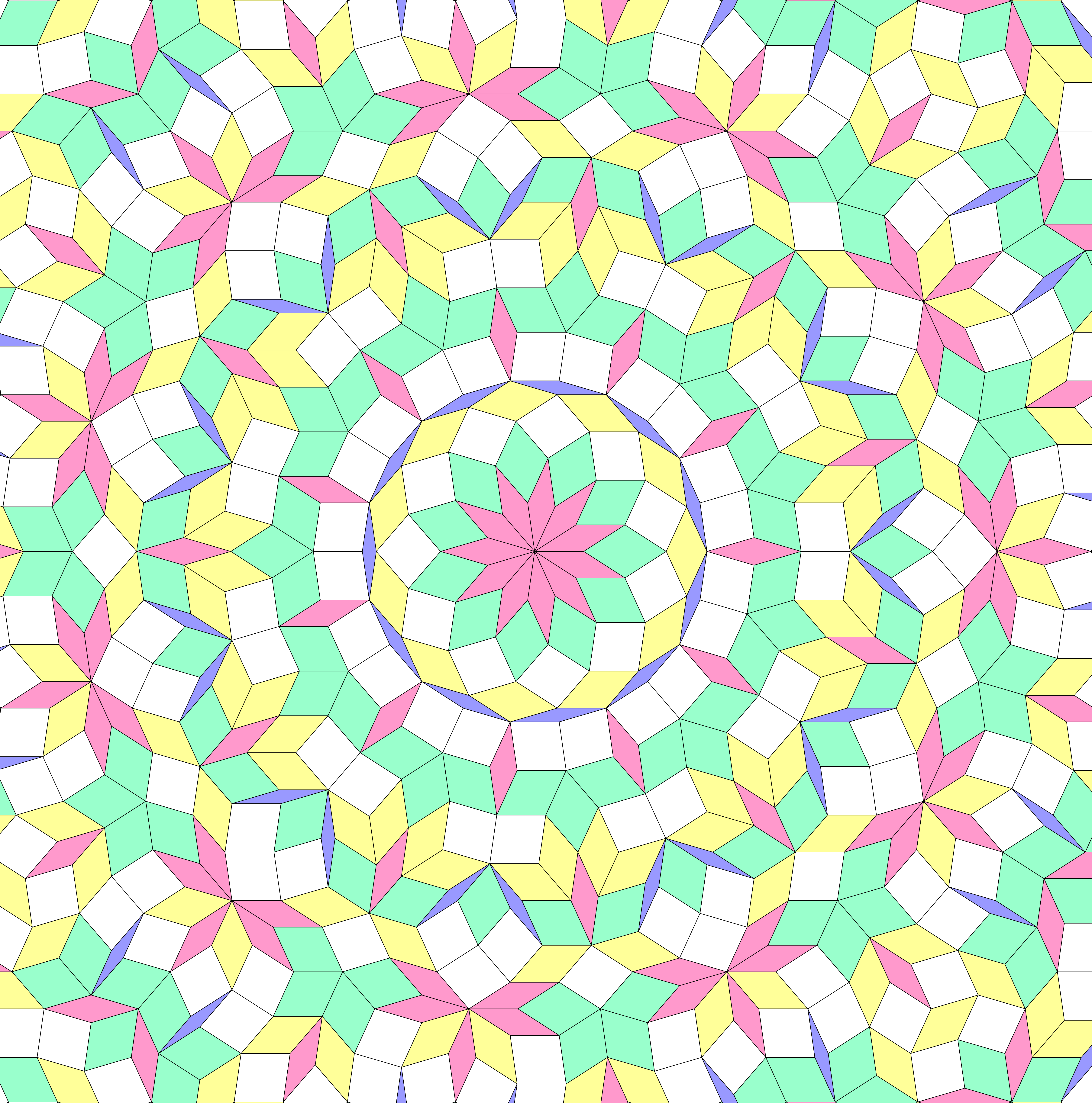}
\caption{$11$-fold : $\dualtiling{11}{\tfrac{1}{11}}$}
\label{fig:P11}
\end{subfigure}
~
\begin{subfigure}[b]{0.3\textwidth}
\includegraphics[width=\textwidth]{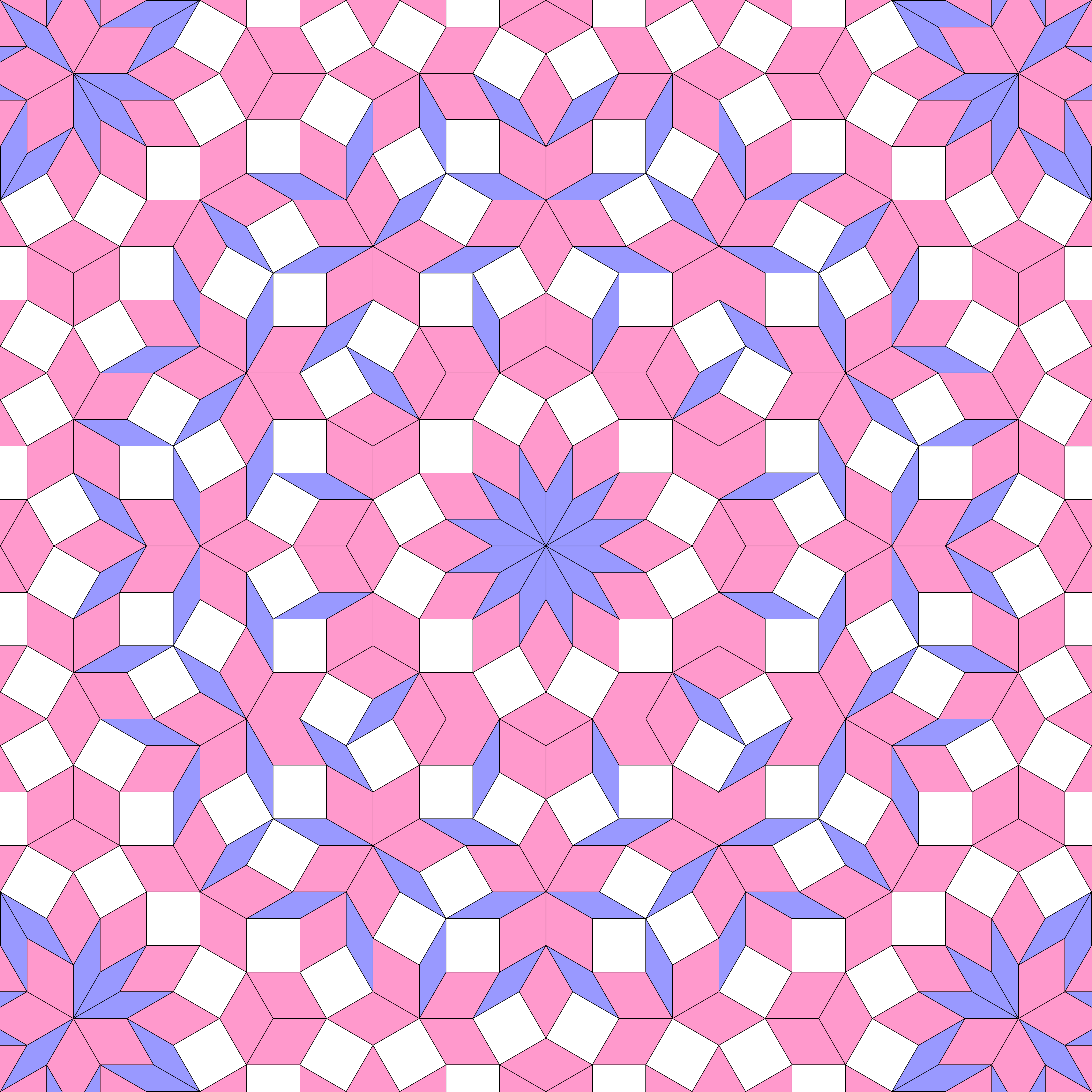}
\caption{$12$-fold : $\dualtiling{6}{\tfrac{1}{2}}$}
\label{fig:P6-2}
\end{subfigure}
\caption{Central patch of the multigrid dual tiling with exactly $n$-fold rotational symmetry for $n\in\{7,8,9,10,11,12\}$}
  \label{fig:nfold_789}
\end{figure}
Let $n\geq 3$ be an integer. By Theorem \ref{thm:regularity} the multigrid $\multigrid{n}{\tfrac{1}{2}}$ and $\multigrid{n}{\tfrac{1}{n}}$ are regular, so their dual tilings $\dualtiling{n}{\tfrac{1}{2}}$ and $\dualtiling{n}{\tfrac{1}{n}}$ are edge-to-edge rhombus tilings. As mentioned above, the multigrid dual tilings are cut-and-project \cite{gahler1986} and therefore also uniformly recurrent.

Let us remark that the dualization process commutes with rotations around the origin, so if a multigrid has some rotational symmetry around the origin then so does its dual tiling.

So for odd $n$, $\dualtiling{n}{\tfrac{1}{n}}$ has global $n$-fold rotational symmetry because the grid also has global $n$-fold rotational symmetry, indeed applying the rotation of angle $\tfrac{2\pi}{n}$ centered on the origin to the multigrid sends $\zeta_n^k$ to $\zeta_n^{k+1}$ and since the offset is the same on all directions the rotated multigrid is the same as the original one.
However this does not hold for even $n$ since in that case we chose $\zeta_n=e^{\imag\frac{\pi}{n}}$ so we have $\zeta_n^{n}= e^{\imag\pi} = -1 = - \zeta_n^0$ so the grid of offset $\tfrac{1}{n}$ for even $n\geq 4$ does not have any rotational symmetry.

Also for odd $n$, $\dualtiling{n}{\tfrac{1}{2}}$ has global $2n$-fold rotational symmetry because the image of $\zeta_n^0$ by the rotation of angle $\tfrac{\pi}{n}$ is \[e^{\imag\tfrac{\pi}{n}} = -e^{\imag\frac{(n+1)\pi}{n}} = -e^{\imag\frac{2\ceil{\frac{n}{2}}\pi}{n}} = - \zeta_n^{\ceil{\frac{n}{2}}}\]  so with $1-\tfrac{1}{2} = \tfrac{1}{2}$ (\ie the offset in direction $\zeta_n^i$ and in its reverse direction $-\zeta_n^i$ is the same) we get global $2n$-fold rotational symmetry for $\multigrid{n}{\tfrac{1}{2}}$ and $\dualtiling{n}{\tfrac{1}{2}}$. For even $n$ we use also the fact that with offset $\tfrac{1}{2}$ we get offset along $\zeta_n^i$ is the same as along the opposite direction $-\zeta_n^i$ together with $\zeta_n=e^{\imag\tfrac{\pi}{n}}$ which means that $\zeta_n^{n+i} = -\zeta_n^i$ to get global $2n$-fold rotational symmetry for $\multigrid{n}{\tfrac{1}{2}}$ and $\dualtiling{n}{\tfrac{1}{2}}$.

When we combine these with the crystallographic restriction which implie that for any $n\geq 4$ these tilings are non-periodic we get Theorem \ref{thm:tilings}.
Remark that for even $n$, $\dualtiling{n}{\tfrac{1}{2}}$ has global $2n$-fold rotational symmetry and $\dualtiling{\frac{n}{2}}{\tfrac{1}{2}}$ has exactly $n$-fold global rotational symmetry. So for any $n$ there exists a tiling with exactly $n$-fold global rotational symetry, see the examples for $n\in\{7,8,9,10,11,12\}$ in Figure \ref{fig:nfold_789}, and for $n=23$ in Figure \ref{fig:P23}.

Remark also that for odd $n$, for any $r\in\mathbb{Q}\cap\openinterval{0}{1}$ the multigrid $\multigrid{n}{r}$ is regular and the multigrid dual tiling $\dualtiling{n}{r}$ has global $n$-fold rotational symmetry, except the specific case $r=\tfrac{1}{2}$ that has global $2n$-fold rotational symmetry. The choice of $r=\tfrac{1}{n}$ is mainly due to the fact that the canonical Penrose rhombus tiling is $\dualtiling{5}{\tfrac{1}{5}}$ so $\dualtiling{n}{\tfrac{1}{n}}$ is in that sense a generalization of the canonical Penrose rhombus tiling.

Note that Theorem \ref{thm:tilings} is stated for $n> 3$ because for $n=3$ the multigrid dual tilings with offset $\tfrac{1}{2}$ and $\tfrac{1}{3}$ are periodic.

\section{Regularity of $n$-fold multigrids and trigonometric equations}
\label{sec:regular_trigonometric}

In this section we present the link between the regularity or singularity of $n$-fold multigrids and some trigonometric equations.

\begin{proposition}[Regularity of multigrids and trigonometric equations]
  \label{prop:regularity}
  Let $n\in\mathbb{N}$, and $\gamma_0,\gamma_1,\dots \gamma_{n-1}$ be offsets in $\closedopeninterval{0}{1}$.
  Assume that for any $p,q$ such that $0<q<p<n$ and any $r_0 \in \mathbb{Z}-\gamma_0,\ r_q \in \mathbb{Z}-\gamma_q$ and $r_p \in \mathbb{Z}-\gamma_p$ we have either Inequation (\ref{equ:trigo1_odd}) when $n$ is odd, or Inequation (\ref{equ:trigo1_even}) when $n$ is even.
  \begin{align}\label{equ:trigo1_odd} (n \text{ odd })\qquad r_0\sin\tfrac{2(p-q)\pi}{n} + r_p\sin\tfrac{2q\pi}{n} - r_q\sin\tfrac{2p\pi}{n} &\neq 0\\
    \label{equ:trigo1_even} (n \text{ even })\qquad r_0\sin\tfrac{(p-q)\pi}{n} + r_p\sin\tfrac{q\pi}{n} - r_q\sin\tfrac{p\pi}{n} &\neq 0\end{align}
    Then the grid $G_n(\gamma_0,\gamma_1,\dots \gamma_{n-1})$ is regular.
\end{proposition}

\begin{proof}
  We will prove this proposition by contradiction \ie we assume a grid is singular and we show that it implies the existance of $r_0, r_p, r_q$ such that the Inequation (\ref{equ:trigo1_odd}) is contradicted if $n$ is odd, and Inequation (\ref{equ:trigo1_even}) is contradicted if $n$ is even.

  We will actually prove it for odd $n$, the proof for even $n$ is exactly the same and it is only needed to replace the formula of $\zeta_n^k$ which in the even case is $e^{\imag\frac{k\pi}{n}}$ instead of $e^{\imag\frac{2k\pi}{n}}$, which means that in the angles we remove the factor 2.
  
Let $n$ be an odd integer and let $\gamma_0,\gamma_1,\dots \gamma_{n-1} \in [0,1[$ such that $G_n(\gamma_0,\gamma_1,\dots \gamma_{n-1})$ is singular.
This means that there exist $z\in\mathbb{C}$ at the intersection of three lines, up to relabeling and rotation we chose to consider it is at the intersection of $H(\zeta_n^0,\gamma_0),\ H(\zeta_n^q, \gamma_q)$ and $H(\zeta_n^p,\gamma_p)$ for some $0<q<p<n$, see Figure \ref{fig:intersection}. This means that there exist $k_0, k_q, k_p \in \mathbb{Z}$ such that
    \[ \begin{cases}
      \text{Re}(z) = k_0 - \gamma_0\\
      \text{Re}(z\cdot \bar{\zeta_n^q}) = k_q - \gamma_q\\
      \text{Re}(z\cdot \bar{\zeta_n^p}) = k_p - \gamma_p
    \end{cases}\]
\begin{figure}[t]
\includegraphics[width=\textwidth]{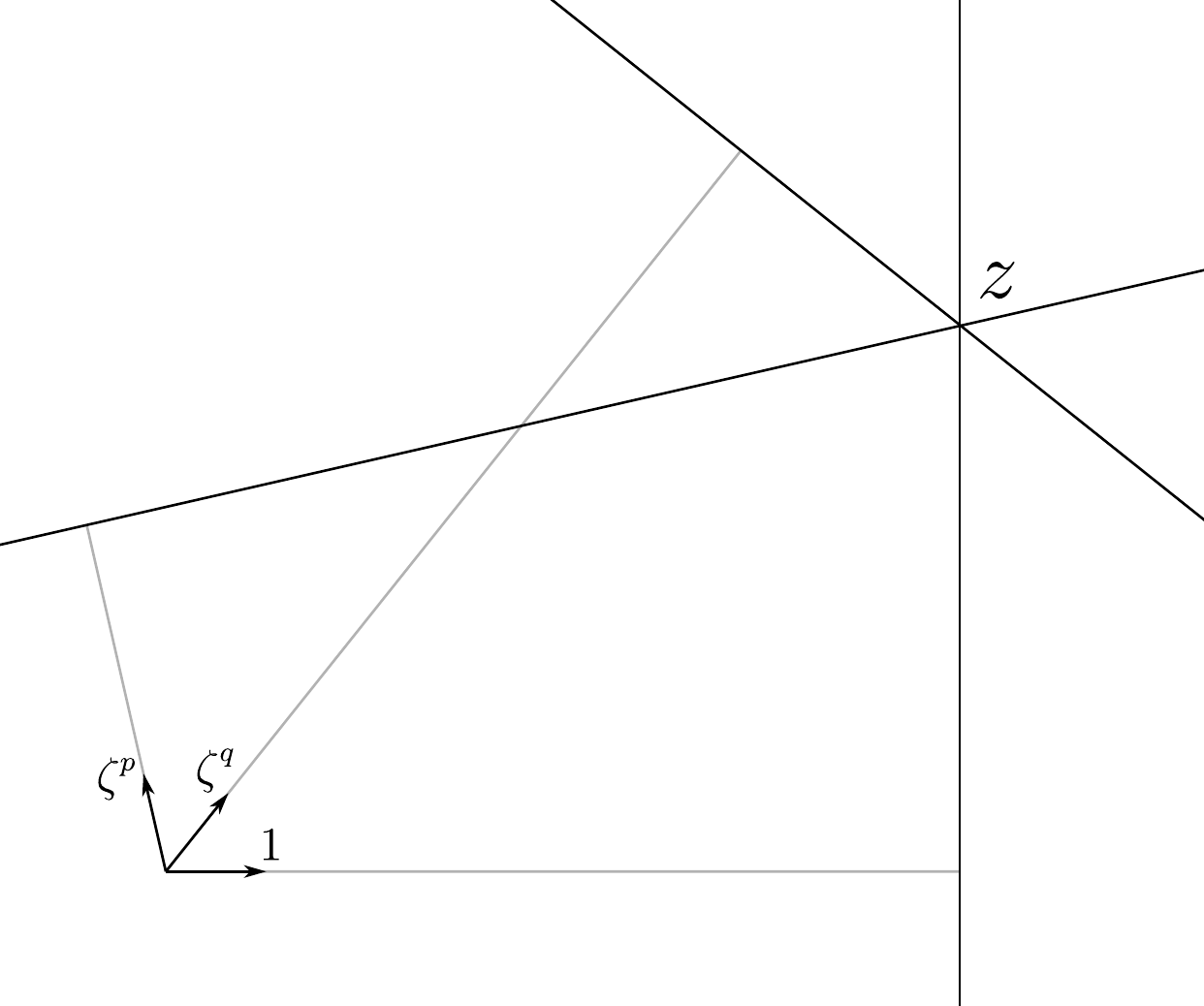}
  \caption{Intersection of three lines}
  \label{fig:intersection}
\end{figure}

Write $z= k_0 - \gamma_0 + iy$ and $\zeta_n=e^{\tfrac{2i\pi}{n}}$. Now we have
\[\begin{cases}
  z = k_0 - \gamma_0 + iy\\
  (k_0 - \gamma_0)\cos\tfrac{2q\pi}{n} + y\sin\tfrac{2q\pi}{n} = k_q - \gamma_q \\
  (k_0 - \gamma_0)\cos\tfrac{2p\pi}{n} + y\sin\tfrac{2p\pi}{n} = k_p - \gamma_p
\end{cases}\]
Let us now cancel out the $y$ terms by $L_3 \leftarrow \sin\tfrac{2p\pi}{n}L_2 - \sin\tfrac{2q\pi}{n}L_3$.
\[\begin{cases}
  z = k_0 - \gamma_0 + iy \\
  (k_0 - \gamma_0)\cos\tfrac{2q\pi}{n} + y\sin\tfrac{2q\pi}{n} = k_q - \gamma_q \\
  (k_0 - \gamma_0)\cos\tfrac{2q\pi}{n}\sin\tfrac{2p\pi}{n}- (k_0 - \gamma_0)\cos\tfrac{2p\pi}{n}\sin\tfrac{2q\pi}{n} = (k_q - \gamma_q)\sin\tfrac{2p\pi}{n} - (k_p - \gamma_p)\sin\tfrac{2q\pi}{n}
\end{cases}\]
Now let us study the third line to simplify it.
\begin{align*}
&(k_0 - \gamma_0)\cos\tfrac{2q\pi}{n}\sin\tfrac{2p\pi}{n} - (k_0 - \gamma_0)\cos\tfrac{2p\pi}{n}\sin\tfrac{2q\pi}{n} = (k_q - \gamma_q)\sin\tfrac{2p\pi}{n} - (k_p - \gamma_p)\sin\tfrac{2q\pi}{n} \\
&\Leftrightarrow(k_0 - \gamma_0)\left(\cos\tfrac{2q\pi}{n}\sin\tfrac{2p\pi}{n} - \cos\tfrac{2p\pi}{n}\sin\tfrac{2q\pi}{n}\right) = (k_q - \gamma_q)\sin\tfrac{2p\pi}{n} - (k_p - \gamma_p)\sin\tfrac{2q\pi}{n}\\
&\Leftrightarrow(k_0 - \gamma_0)\sin\tfrac{2(p-q)\pi}{n} + (k_p +\gamma_p)\sin\tfrac{2q\pi}{n} - (k_q - \gamma_q)\sin\tfrac{2p\pi}{n} = 0\\
\end{align*}
If we rewrite $k_0 - \gamma_0,\ k_p - \gamma_p,\ k_q - \gamma_q$ as $r_0,\ r_p,\ r_q$ we obtain
\[ r_0\sin\tfrac{2(p-q)\pi}{n} + r_p\sin\tfrac{2q\pi}{n} - r_q\sin\tfrac{2p\pi}{n} = 0 \]
which is exactly the contradiction of Inequation (\ref{equ:trigo1_odd}).
\end{proof}

In the next section we consider the solutions to these kind of trigonometric equations.

\section{Trigonometric diophantine equations}
\label{sec:trigonometric_diophantine}

We call \emph{rational angles} the set $\pi\mathbb{Q}$.
We consider now equations of the type
\begin{align} \label{equ:cos_general} A\cos(a) + B\cos(b) + C\cos(c) = 0 \end{align}
with $a,b$ and $c$ rational angles.

In the previous paragraph we had sine instead of cosine but we can always convert sine to cosine, and we had $A\in \mathbb{Z}-\gamma$ for some real number $0\leq \gamma <1$ and similarly for $B$ and $C$ but now we will consider $A,B$ and $C$ to be rationals.

\begin{theorem}[ConwayJones, '76]
  \label{thm:conwayjones}
  Suppose we have at most four distinct rational angles strictly between 0 and $\tfrac{\pi}{2}$ for which some rational linear combination of their cosines has rational value but no proper subset has this property.\\
  Then the appropriate linear combination is proportional to one from the following list:
  \begin{align}
    & \cos\left(\tfrac{\pi}{3}\right) = \tfrac{1}{2} \label{equ:th1}\\
    & -\cos(\phi) + \cos\left(\tfrac{\pi}{3}-\phi\right) + \cos\left(\tfrac{\pi}{3}+\phi \right) = 0 \ (0<\phi<\tfrac{\pi}{6}) \label{equ:th2}\\
    & \cos\left(\tfrac{\pi}{5}\right) - \cos\left(\tfrac{2\pi}{5}\right) = \tfrac{1}{2} \label{equ:th3}\\
    & \cos\left(\tfrac{\pi}{7}\right) - \cos\left(\tfrac{2\pi}{7}\right) + \cos\left(\tfrac{3\pi}{7}\right) = \tfrac{1}{2}\\
    & \cos\left(\tfrac{\pi}{5}\right) - \cos\left(\tfrac{\pi}{15}\right) + \cos\left(\tfrac{4\pi}{15}\right) = \tfrac{1}{2}\\
    & -\cos\left(\tfrac{2\pi}{5}\right) + \cos\left(\tfrac{2\pi}{15}\right) - \cos\left(\tfrac{7\pi}{15}\right) = \tfrac{1}{2}\\
    & \cos\left(\tfrac{\pi}{7}\right) + \cos\left(\tfrac{3\pi}{7}\right) - \cos\left(\tfrac{\pi}{21}\right) + \cos\left(\tfrac{8\pi}{21}\right) = \tfrac{1}{2}\\
    & \cos\left(\tfrac{\pi}{7}\right) - \cos\left(\tfrac{2\pi}{7}\right) + \cos\left(\tfrac{2\pi}{21}\right) - \cos\left(\tfrac{5\pi}{21}\right) = \tfrac{1}{2}\\
    & -\cos\left(\tfrac{2\pi}{7}\right) + \cos\left(\tfrac{3\pi}{7}\right) + \cos\left(\tfrac{4\pi}{21}\right) + \cos\left(\tfrac{10\pi}{21}\right) = \tfrac{1}{2} \\
    & -\cos\left(\tfrac{\pi}{15}\right) + \cos\left(\tfrac{2\pi}{15}\right) + \cos\left(\tfrac{4\pi}{15}\right) - \cos\left(\tfrac{7\pi}{15}\right) = \tfrac{1}{2}
  \end{align}
\end{theorem}
See the original article \cite{conway1976} for the proof.
The proof is based on a more general result on vanishing sums of roots of unity. And this is proved using complex numbers and the theory of vanishing formal sums. If we adapt this result for sum of three cosines that have value zero we get.

\begin{corollary}
\label{cor:tricosine}
  Let $a \leq b \leq c$ be rational angles strictly between $0$ and $\tfrac{\pi}{2}$ and not all equal, and let $A,\ B,\ C$ be non-zero rationals.
  
  If $ A\cos(a) + B\cos(b) + C\cos(c) = 0$
  then either
  \[\begin{cases}
    a=\tfrac{\pi}{5}\\
    b=\tfrac{\pi}{3}\\
    c=\tfrac{2\pi}{5}\\
    B = C = -A\\
  \end{cases}\ \text{or} \qquad \begin{cases}
  0<a<\tfrac{\pi}{6}\\
  b = \tfrac{\pi}{3}-a\\
  c = \tfrac{\pi}{3}+a\\
  B = C = -A\\
  \end{cases}\]
\end{corollary}

\begin{proof}
  We just need to  apply Theorem \ref{thm:conwayjones}.
  First remark that there is no solution for $A\cos(a) + B\cos(b) = 0$ with $a$ and $b$ distinct and strictly between $0$ and $\tfrac{\pi}{2}$, and $A$ and $B$ non zero.
  Now with $0<a<b<c<\tfrac{\pi}{2}$, we have either a combination of Equations (\ref{equ:th1}) and (\ref{equ:th3}) (first case) or Equation (\ref{equ:th2}) (second case).
\end{proof}

\section{Proof of Theorem \ref{thm:regularity}}
\label{sec:proof}

Here we use Proposition \ref{prop:regularity} and Corollary \ref{cor:tricosine} to prove Theorem \ref{thm:regularity}.

\subsection{For odd $n$}
\label{subsec:odd}
First let us remark that in Theorem \ref{thm:regularity} the first statement when restricted to odd $n$ is a strict subcase of the second statement, so here we will prove the second statement which is as follows:  
for any odd $n\geq 3$ and any tuple of non-zero rational offsets $\gamma = (\gamma_i)_{0\leq i < 0} \in \left(\mathbb{Q}\cap \openinterval{0}{1}\right)^n$ the $n$-fold multigrid $\multigrid{n}{\gamma}$ is regular. 
We reformulate this with Proposition \ref{prop:regularity} as :
for any odd $n\geq 3$, for any $0< p < q <n$ and any three non-zero rational offsets  $\gamma_0, \gamma_p, \gamma_q$, for any $r_0 \in \mathbb{Z}-\gamma_0$, $r_p \in \mathbb{Z} - \gamma_p$ and $r_q \in \mathbb{Z} - \gamma_q$ we have
\[r_0\sin\tfrac{2(p-q)\pi}{n} + r_p\sin\tfrac{2q\pi}{n} - r_q\sin\tfrac{2p\pi}{n} \neq 0.\]
Actually we prove a slitghly reformulated version:
for any odd $n\geq 3$, for any $0< p < q <n$ and any three non-zero rationals $r_0,r_p, r_q\in\mathbb{Q}_{\backslash \{0\}}$ we have \[r_0\sin\tfrac{2(p-q)\pi}{n} + r_p\sin\tfrac{2q\pi}{n} - r_q\sin\tfrac{2p\pi}{n} \neq 0.\]
To apply Corollary \ref{cor:tricosine} we first need to translate the formula with sine and with angles in $\closedopeninterval{0}{2\pi}$ as a formula with cosine and angles in $\openinterval{0}{\tfrac{\pi}{2}}$.

\begin{lemma}[Sine and Cosine]
  \label{lemma:sincos}
  For $\theta \in \closedopeninterval{0}{2\pi}$ we have
  \[ \sin(\theta) = (-1)^{\floor{\frac{\theta}{\pi}}} \cos\left( (-1)^{\floor{\frac{2\theta}{\pi}}}( \floor{\tfrac{\theta}{\pi}} \pi + \tfrac{\pi}{2} - \theta)\right) \]
  and $(-1)^{\floor{\frac{2\theta}{\pi}}}( \floor{\tfrac{\theta}{\pi}} \pi + \tfrac{\pi}{2} - \theta) \in \closedinterval{0}{\tfrac{\pi}{2}}$.

\end{lemma}
\begin{proof}
  This result is just a rewriting of :
  \begin{itemize}
  \item if $0\leq \theta < \tfrac{\pi}{2}$ then $\sin(\theta) = \cos(\tfrac{\pi}{2} - \theta)$ and $(\tfrac{\pi}{2} - \theta) \in \closedinterval{0}{\tfrac{\pi}{2}}$
  \item if $\tfrac{\pi}{2} \leq \theta < \pi$ then $\sin(\theta) = \cos(\theta - \tfrac{\pi}{2})$ and $(\theta - \tfrac{\pi}{2}) \in \closedinterval{0}{\tfrac{\pi}{2}}$
  \item if $\pi \leq \theta < \tfrac{3\pi}{2}$ then $\sin(\theta) = -\cos(\tfrac{3\pi}{2}-\theta)$ and $(\tfrac{3\pi}{2}-\theta) \in \closedinterval{0}{\tfrac{\pi}{2}}$
  \item if $\tfrac{3\pi}{2} \leq \theta < 2\pi$ then $\sin(\theta) = -\cos(\theta - \tfrac{3\pi}{2})$ and $(\theta - \tfrac{3\pi}{2}) \in \closedinterval{0}{\tfrac{\pi}{2}}$
  \end{itemize}
\end{proof}
We define $\epsilon(\theta) := (-1)^{\floor{\frac{\theta}{\pi}}}$ and $\phi(\theta) := (-1)^{\floor{\frac{2\theta}{\pi}}}( \floor{\tfrac{\theta}{\pi}} \pi + \tfrac{\pi}{2} - \theta)$.
Remark that this means that $\theta = \floor{\tfrac{\theta}{\pi}}\pi + \tfrac{\pi}{2} -(-1)^{\floor{\frac{2\theta}{\pi}}}\phi(\theta)$.

Let $n,p,q$ be integers such that $n$ is odd,  $n\geq 3$ and $0< q < p < n$.
By contradiction suppose that there exists $r_0$,$r_p$ and $r_q$ non-zero rationals such that
\[ r_0\sin(\tfrac{2(p-q)\pi}{n}) + r_p\sin(\tfrac{2q\pi}{n}) - r_q\sin(\tfrac{2p\pi}{n}) = 0. \]
By Lemma \ref{lemma:sincos} we have
\[ r_0\epsilon(\tfrac{2(p-q)\pi}{n})\cos(\phi(\tfrac{2(p-q)\pi}{n})) + r_p\epsilon(\tfrac{2q\pi}{n})\cos(\phi(\tfrac{2q\pi}{n})) - r_q\epsilon(\tfrac{2p\pi}{n})\cos(\phi(\tfrac{2p\pi}{n})) = 0. \]
Which we reformulate as
\[ r_0' \cos(\theta_0) + r_p'\cos(\theta_q) + r_q'\cos(\theta_p) = 0\]
with $r_0' := r_0\epsilon(\tfrac{2(p-q)\pi}{n})$, $r_p' := r_p\epsilon(\tfrac{2q\pi}{n})$, $r_q' := - r_q\epsilon(\tfrac{2p\pi}{n})$ and $\theta_0 := \phi(\tfrac{2(p-q)\pi}{n})$, $\theta_q := \phi(\tfrac{2q\pi}{n})$, $\theta_p := \phi(\tfrac{2p\pi}{n})$.

Remark that for odd $n$ and any $0<k<n$ we have $\tfrac{2k\pi}{n}\notin\{0,\tfrac{\pi}{2},\pi, \tfrac{3\pi}{2}\}$, so $\phi(\tfrac{2k\pi}{n})\notin\{0,\tfrac{\pi}{2}\}$. This implies that $0< \theta_0, \theta_p, \theta_q <\tfrac{\pi}{2}$.

Moreover for odd $n$ we have that $\theta_0,\theta_p,\theta_q$ are not all equal. By contradction if $\theta_0 = \theta_p = \theta_p$ we have that $\tfrac{2p\pi}{n},\tfrac{2q\pi}{n}, \tfrac{2(p-q)\pi}{n} \in \phi^{-1}(\{\theta_0\}) = \{ \tfrac{\pi}{2} - \theta_0, \tfrac{\pi}{2} + \theta_0, \tfrac{3\pi}{2} - \theta_0, \tfrac{3\pi}{2} + \theta_0\}$ which is impossible.
So we have \[ r_0' \cos(\theta_0) + r_p'\cos(\theta_q) + r_q'\cos(\theta_p) = 0\] with non-zero rationals $r_0', r_p', r_q'$ and with three angles strictly between $0$ and $\tfrac{\pi}{2}$ and not all equal.

So we can apply Corollary \ref{cor:tricosine} and we now have two cases:
\begin{enumerate}
\item $\{\theta_0, \theta_p, \theta_q\} = \{ \tfrac{\pi}{5}, \tfrac{\pi}{3}, \tfrac{2\pi}{5} \}$
\item $\{\theta_0, \theta_p, \theta_q\} = \{ \theta, \tfrac{\pi}{3}- \theta, \tfrac{\pi}{3}+\theta\}$ for some $0< \theta < \tfrac{\pi}{6}$
\end{enumerate}
Let us now show that both cases lead to a contradiction.
In the first case we have that $\{\tfrac{2p\pi}{n},\tfrac{2q\pi}{n},\tfrac{2(p-q)\pi}{n}\} = \{ \theta_1, \theta_2, \theta_3\}$ with
\begin{align*}
  \theta_1 &\in \phi^{-1}(\{\tfrac{\pi}{5}\}) = \{\tfrac{3\pi}{10}, \tfrac{7\pi}{10}, \tfrac{13\pi}{10}, \tfrac{17\pi}{10}\} \\
  \theta_2 &\in \phi^{-1}(\{\tfrac{\pi}{3}\}) = \{\tfrac{\pi}{6}, \tfrac{5\pi}{6}, \tfrac{7\pi}{6}, \tfrac{11\pi}{6}\} \\
  \theta_3 &\in \phi^{-1}(\{\tfrac{2\pi}{5}\}) = \{\tfrac{\pi}{10}, \tfrac{9\pi}{10}, \tfrac{11\pi}{10}, \tfrac{19\pi}{10}\}
\end{align*}
This is impossible because by definition $\tfrac{2p\pi}{n} = \tfrac{2q\pi}{n} + \tfrac{2(p-q)\pi}{n}$ and we have no $\theta_1,\theta_2,\theta_3$ as defined above such that one is the sum of the two other.

In the second case we have that $\{\theta_0, \theta_p, \theta_q\} = \{\theta, \tfrac{\pi}{3}-\theta, \tfrac{\pi}{3}-\theta\}$ for some $0< \theta < \tfrac{\pi}{6}$.
Now we use $\tfrac{2p\pi}{n} = \tfrac{2q\pi}{n} + \tfrac{2(p-q)\pi}{n}$ and by definition we have
\begin{align*}
  \tfrac{2p\pi}{n} &= \floor{\tfrac{2p}{n}}\pi + \tfrac{\pi}{2} -(-1)^{\floor{\frac{4p}{n}}} \theta_p = \floor{\tfrac{2p}{n}}\pi + \tfrac{\pi}{2} \pm \theta_p\\
  \tfrac{2q\pi}{n} &= \floor{\tfrac{2q}{n}}\pi + \tfrac{\pi}{2} -(-1)^{\floor{\frac{4q}{n}}} \theta_q = \floor{\tfrac{2q}{n}}\pi + \tfrac{\pi}{2} \pm \theta_q\\
  \tfrac{2(p-q)\pi}{n} &= \floor{\tfrac{2(p-q)}{n}}\pi + \tfrac{\pi}{2} -(-1)^{\floor{\frac{4(p-q)}{n}}} \theta_0 = \floor{\tfrac{2(p-q)}{n}}\pi + \tfrac{\pi}{2} \pm \theta_0
\end{align*}
By assembling these two we get
\[ (\floor{\tfrac{2p}{n}} - \floor{\tfrac{2q}{n}}  - \floor{\tfrac{2(p-q)}{n}} -1)\pi + \tfrac{\pi}{2} = \pm \theta_p \pm \theta_0 \pm\theta_q\]
And with $\{\theta_0, \theta_p, \theta_q\} = \{\theta, \tfrac{\pi}{3}-\theta, \tfrac{\pi}{3}-\theta\}$ we get 
\[ (\pm \theta_p \pm \theta_0 \pm\theta_q) \in \{ \pm 3 \theta, \pm \theta, \pm\tfrac{2\pi}{3} \pm \theta\}\]
However this is impossible since for $0<\theta < \tfrac{\pi}{6}$, we have
\[\left(\mathbb{Z}\pi + \tfrac{\pi}{2}\right)\cap \{ \pm 3 \theta, \pm \theta, \pm\tfrac{2\pi}{3} \pm \theta\} = \emptyset\]

By contradiction we proved that for odd $n$, any $n$-fold multigrid with non-zero rational offsets is regular.

\subsection{For even $n$}

Let us now prove the first statement of Theorem \ref{thm:regularity} for even $n$ which is:
for any even $n\geq 4$ and any non-zero rational offset $r\in\mathbb{Q}\cap \openinterval{0}{1}$ the $n$-fold multigrid $\multigrid{n}{r}$ is regular. We reformulate it using Proposition \ref{prop:regularity} as:
for any even $n\geq 4$ and any non-zero rational offset $r\in\mathbb{Q}\cap \openinterval{0}{1}$, for any $0<q<p<n$ and any $r_0 \in \mathbb{Z}-r$, $r_p\in\mathbb{Z}-r$ and $r_q \in \mathbb{Z}-r$ we have \[r_0\sin\tfrac{(p-q)\pi}{n} + r_p\sin\tfrac{q\pi}{n} - r_q\sin\tfrac{p\pi}{n} \neq 0 .\]
We will prove this by contradiction.
Let $n\geq 4$ be an even integer and $r$ be a non-zero rational offset.
Suppose that there exists $p,q,r_0,r_p,r_q$ with $p,q$ integers, $0<q<p<n$ and  $r_0 \in \mathbb{Z}-r$, $r_p\in\mathbb{Z}-r$ and $r_q \in \mathbb{Z}-r$, such that \[r_0\sin\tfrac{(p-q)\pi}{n} + r_p\sin\tfrac{q\pi}{n} - r_q\sin\tfrac{p\pi}{n} = 0 .\]
We apply Lemma \ref{lemma:sincos} with the fact that since $0< \tfrac{p\pi}{n}, \tfrac{q\pi}{n}, \tfrac{(p-q)\pi}{n}<\pi$ we have $\epsilon(\tfrac{k\pi}{n}) = 1$ and $\phi(\tfrac{k\pi}{n}) = (-1)^{\floor{\frac{2k}{n}}}( \tfrac{\pi}{2} - \theta)$ for $k \in\{p,q,(p-q)\}$.
We obtain
\[ r_0\cos\theta_0 + r_p\cos\theta_q - r_q\cos\theta_p = 0\]
with $\theta_0 = \phi(\tfrac{(p-q)\pi}{n})$, $\theta_p = \phi(\tfrac{p\pi}{n})$ and $\theta_q = \phi( \tfrac{q\pi}{n})$.
And since $0< \tfrac{p\pi}{n}, \tfrac{q\pi}{n}, \tfrac{(p-q)\pi}{n}<\pi$ we have $\theta_0,\theta_p,\theta_q \in \closedopeninterval{0}{\tfrac{\pi}{2}}$. In particular since $n$ is even we can have $\tfrac{p\pi}{n} = \tfrac{\pi}{2}$ which means that we can have $\theta_p = 0$ (and also for $\theta_0$ or $\theta_q$).
Note also that with $n$ even (contrary to the odd case) we can have $\theta_0 = \theta_p = \theta_q$, for example with $n=6$, $q=2$ and $p=4$ we have $\theta_0 = \theta_p = \theta_q = \tfrac{\pi}{6}$.
Which means that now we have a disjunction of four cases:
\begin{enumerate}
\item $\theta_0 = \theta_p = \theta_q$
\item $0< \theta_0, \theta_p, \theta_q < \tfrac{\pi}{2}$ and not all equal
\item two of the angles are 0 and the other one is not
\item one of the angles is 0 and the other two are not
\end{enumerate}

The first case reduces to $(r_0 + r_p -r_q)\cos\theta_0 = 0$ and with $\theta_0\in\closedopeninterval{0}{\tfrac{\pi}{2}}$ we have $\cos\theta_0 \neq 0$ so $(r_0 + r_p -r_q) = 0$ but this is impossible because $(r_0 + r_p -r_q) \in \mathbb{Z} - r$ and $0\notin (\mathbb{Z}-r)$ for $r\in(\mathbb{Q}\cap \openinterval{0}{1})$.

The second case is the same as the one discussed in Subsection \ref{subsec:odd} above, the main thing we used in the proof for odd $n$ is the fact that $\tfrac{2p\pi}{n} = \tfrac{2q\pi}{n} + \tfrac{2(p-q)\pi}{n}$ but we have the same for even $n$ with  $\tfrac{p\pi}{n} = \tfrac{q\pi}{n} + \tfrac{(p-q)\pi}{n}$. So the proof holds and this case is impossible.

The third case is impossible because for two angles to be 0, we need two of $\{\tfrac{p\pi}{n}, \tfrac{q\pi}{n}, \tfrac{(p-q)\pi}{n}\}$ to be $\tfrac{\pi}{2}$ but this is impossible with $0<q<p<n$.

The fourth case reduces to $A\cos a + B\cos b = C$ with $A,B,C \in \mathbb{Z}-r\subset \mathbb{Q}$ so we can apply Theorem \ref{thm:conwayjones} and we get that either $a=b=\tfrac{\pi}{3}$ or $a=\tfrac{\pi}{5}$ and $b=\tfrac{\pi}{5}$. The first subcase is impossible because $\{\theta_0,\theta_p,\theta_q\} = \{0,\tfrac{\pi}{3}\}$ implies $\{\tfrac{p\pi}{n},\tfrac{q\pi}{n}, \tfrac{(p-q)\pi}{n}\} \subseteq \{ \tfrac{\pi}{6},\tfrac{\pi}{2}, \tfrac{5\pi}{6}\}$ and this is incompatible with $\tfrac{p\pi}{n} = \tfrac{q\pi}{n} + \tfrac{(p-q)\pi}{n}$. The second subcase is also impossible for the same reason as (up to interchanging $\theta_0$, $\theta_p$ and $\theta_q$) we would have $\theta_0 = 0$, $\theta_p=\tfrac{\pi}{5}$ and $\theta_q = \tfrac{2\pi}{5}$ which means that $\tfrac{(p-q)\pi}{n} = \tfrac{\pi}{2}$, $\tfrac{p\pi}{n} \in \{ \tfrac{3\pi}{10}, \tfrac{7\pi}{10}\}$ and $\tfrac{q\pi}{n}\in \{ \tfrac{\pi}{10},\tfrac{9\pi}{10}\}$ and this is incompatible with $\tfrac{p\pi}{n} = \tfrac{q\pi}{n} + \tfrac{(p-q)\pi}{n}$.

Note that only in the first case we use the fact that the offset are all the same $r\in\mathbb{Q}\cap\openinterval{0}{1}$, the three other case work for any non-zero rational offsets as for the odd case. And actually we could refine the condition because what is important here is that $0\notin (\mathbb{Z}-\gamma_0-\gamma_q + \gamma_p)$ with $\gamma_0$, $\gamma_p$ and $\gamma_q$ the rational offsets. So for even $n$ if we have a tuple of rational offsets $\gamma = (\gamma_i)_{0\leq i < n}$ such that for any distinct $i$, $j$, $k$ we have $\gamma_i - \gamma_j - \gamma_k\neq 0$ then the multigrid $\multigrid{n}{\gamma}$ is regular.

\begin{figure}[t]
  \includegraphics[width=\textwidth]{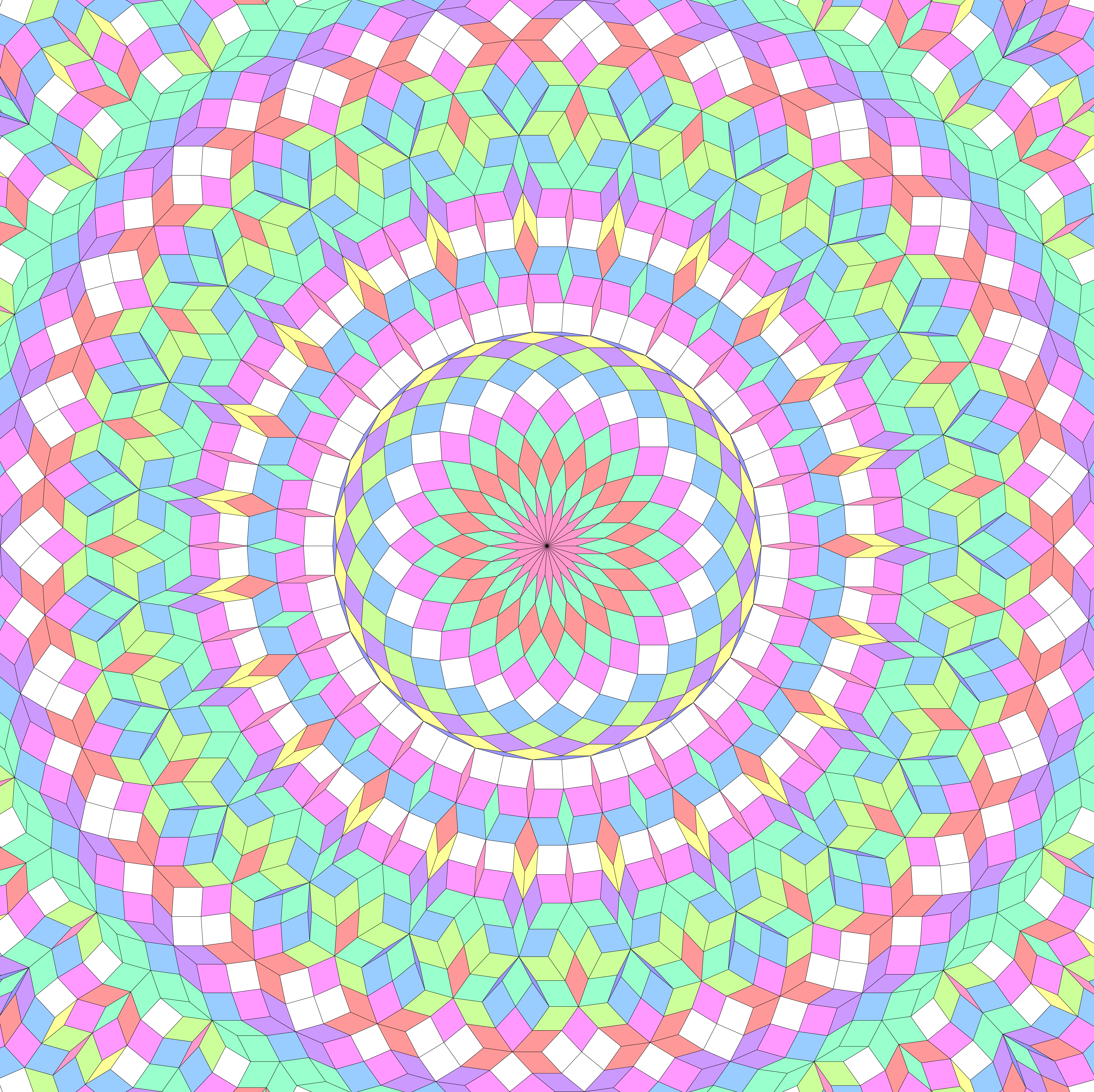}
  \caption{ A central patch of the multigrid dual tiling $\dualtiling{23}{\tfrac{1}{23}}$ with 23-fold rotational symmetry}
  \label{fig:P23}
\end{figure}

\bibliographystyle{alpha}
\bibliography{multigrid_lutfalla}

\end{document}